\documentclass{article}
\usepackage[T1]{fontenc}
\usepackage{graphicx}
\usepackage{amsmath} 
\usepackage{amsfonts}
\usepackage{amssymb}
\usepackage{amsthm}
\usepackage{xcolor}
\usepackage{multirow}
\usepackage{hyperref}
\newtheorem{thrm}{Theorem}
\newtheorem{rmrk}{Remark}
\newtheorem{prpstn}{Proposition}

\begin{document}

\title{A mathematical perspective on the paradox that chemotherapy sometimes works backwards}

\author{\large{Luis A. Fernández, Isabel Lasheras and Cecilia Pola} \\ 
Department of Mathematics, Statistics and Computation,
\\ University of Cantabria, \\
Avda. de los Castros, s/n, 39005 Santander, Spain.
\footnote{lafernandez@unican.es, isabellasheras02@gmail.com, polac@unican.es}}

\date{10/03/2025}
\maketitle

\begin{abstract}
Doctors are well aware that sometimes cancer treatments not only fail, but even work backwards, i.e. they make the treated tumor grow. In this work we present a mathematical perspective on this paradox in the case of chemotherapy, by studying a minimally parameterized mathematical model for the system composed of the tumor and the surrounding vasculature. To this end, we will use a system of two well-established nonlinear ordinary differential equations, which incorporates the cytotoxic (via the Norton-Simon hypothesis) and antiangiogenic effects of chemotherapy. Finally, we provide two theoretical ways to avoid these anomalies.
\end{abstract}



Norton-Simon hypothesis.

\section*{Introduction}
It is well-known that chemotherapy, surgery and radiotherapy are the main types of treatment in oncology. What is less known (and not yet well understood) is that all of them can sometimes end up having the opposite effect to the desired one. This fact has been reported in the medical literature and its cause (although not entirely clear) is attributed by some authors to the effect of the cellular debris (``the resulting dead cell fragments that have yet to be cleared by the immune system") produced by these treatments. As pointed out in \cite{Haak-etal-2021}: ``every attempt to kill cancer cells and cause tumor cell death is a double-edged sword as the resulting unresolved debris stimulates the growth of surviving tumor cells", ``the debris produced by cytotoxic cancer therapy can also contribute to a tumor microenvironment that promotes tumor progression and recurrence", ``chemotherapy, surgery, and radiation all contribute to debris-stimulated cancer. The debris generated from these therapies secrete pro-tumorigenesis factors to promote tumor angiogenesis, proliferation, growth, metastasis, and recurrence". This paradox has been reported also in other recent references, such as \cite{Dalterioetal2020}: ``Accumulating pre-clinical evidence suggests in fact that chemotherapy induces intra-tumoral and systemic changes that can paradoxically promote cancer cell survival/proliferation ultimately fostering dissemination to distant organs". Similar comments can be found in \cite{Chang-etal-2019} and \cite{Sulciner-etal-2018}. 

In this work, we present a mathematical perspective on this issue, by studying a well-established mathematical model for the system composed of the tumor and the surrounding vasculature, \cite{Hahnfeldt-etal-1999}. This system consists of two ordinary differential equations (ODE): the Gompertz ODE for the evolution of tumor volume and the ODE by Hahnfeldt et al. for the evolution of the effective vascular support. We incorporate the effects of cytotoxic and antiangiogenic treatments, using the Norton-Simon hypothesis for the first one (\cite{Norton-Simon1977, Simon-Norton2006}) with Emax type pharmacodynamics and the models by Hahnfeldt et al. for the second one (see \cite{Hahnfeldt-etal-1999, Poleszczuk-etal-2015}). From the full model, that it is minimally parameterized compared to others in the literature (see for instance \cite{Bodzioch-etal24}), we deduce that:

\begin{itemize}
\item[1)] Cytotoxic chemotherapy transforms the tumor-vasculature system into an unstable one when it is applied at a sufficiently high concentration. Then, in normal cases (the majority), tumors shrink in size or slow down their growth, as is well known, but there are also abnormal situations whose volume (paradoxically) increases faster with cytotoxic chemotherapy treatment than without any treatment. 

\item[2)] The set of abnormal cases varies with the treatment used and seems to be small, but expands with the effect of the cytotoxic chemotherapy treatment. It is mainly located around the area where the tumor volume is very large and the vasculature is small. The boundary between this and the set of normal cases is quite blurred, but it can be approximately determined in the model by numerical techniques. It seems crucial to distinguish whether a situation is abnormal before applying a given cytotoxic chemotherapy treatment.

\item[3)] Fortunately, there are some ways to avoid the abnormal situations. One is simply to delay the onset of treatment. If this is possible in the practice, the size of the tumor vasculature should be monitored until it reaches an adequate size. To our knowledge this strategy has not been reported in the literature, but some references suggest that monitoring the size of the vasculature is an important tool in oncology: ``Our study encourages the measure of tumor vasculature as a surrogate for tumor carrying capacity as a biomarker, which may ultimately lead to better-informed patient-specific synergizing of cytotoxic and antiangiogenic treatment", see \cite{Poleszczuk-etal-2015}.
A second procedure is to apply an antiangiogenic treatment prior to achieve a transitory normalization of the vasculature in order to be able to apply a cytotoxic treatment afterwards. This technique has been reported in various medical sources (see \cite{Jain2001, Jain2005, Jain-etal-2011}) and our system also reflects this fact. 

We will illustrate the overall results with numerical simulations using parameter values taken from the specialized literature, adapted to the case of Bevacizumab as regards antiangiogenic treatment.
\end{itemize}

The origin of this work was the academic project \cite{TFG} submitted by the second author under the supervision of the first author with the help of the third one. This is a completely revised version, including new results. 

\section{Pharmacodynamics}

Let us start by introducing the mathematical model proposed by Hahnfeldt and collaborators \cite{Hahnfeldt-etal-1999} to represent the macroscopic evolution of tumor-vasculature system before and after applying different therapies. Denoting by $V(t)$ the volume of the tumor and by $K(t)$ the volume of the surrounding vasculature (effective vascular support) at time $t$, tumor growth is assumed to be governed by the Gompertz law with the variable carrying capacity $K(t)$:

\begin{equation} 
V'(t) = - \lambda_1 V(t) \log \left( \frac{V(t)}{K(t)} \right), 
\label{EV1}
\end{equation}
where the positive parameter $\lambda_1$ denotes the tumor growth rate. To model the evolution of $K(t)$ the kinetics of angiogenic stimulation and inhibition due to the presence of the tumor has to be taken into account:

\begin{equation} 
K'(t) =  - \lambda_2 K(t) + b V(t) - d K(t)V^{2/3}(t).
\label{EK1}
\end{equation}

All the parameters have a physical meaning and are nonnegative:  $\lambda_2$ is related with the spontaneous loss of vasculature, $b$ with the stimulation due to factors secreted by the tumor
and $d$ with the endogenous inhibition. Parameter values of $\lambda_1, \lambda_2, b$ and $d$ have been estimated (see  \cite{Hahnfeldt-etal-1999} and \cite{Poleszczuk-etal-2015}) to describe average tumor growth for different cell tumors. 

In the literature, there are several ways to add the effect $E_c$ of cytotoxic treatments (those that kill cells, especially cancer cells, \cite{NCI}) to the Gompertz ODE. Historically, the ``log-kill hypothesis" was the first one (\cite{Norton-2014}). It states that a given dose of chemotherapy kills the same fraction of tumor cells regardless of the size of the tumor at the time of treatment. In mathematical terms, it can be written as

\begin{equation} 
V'(t) = - \lambda_1 V(t) \log \left( \frac{V(t)}{K(t)} \right)-E_c(t)V(t).
\label{EV2}
\end{equation}

Later on, it was found experimentally (see \cite{Norton-Simon1977})) that the proportion of cells killed by an effective therapy was not independent of tumor size, and this led to the ``Norton-Simon hypothesis": for a given drug and a given dose level, the cell kill is proportional to the growth rate that would be expected for an unperturbed tumor of that size. Expressed by formulas,

\begin{equation} 
V'(t) = - \lambda_1 V(t) \log \left( \frac{V(t)}{K(t)} \right)(1-E_c(t)).
\label{EV3}
\end{equation}

Several problems concerning the optimization of cytotoxic chemotherapy treatments related to this model (in the particular case where $K(t)$ is constant) were studied in \cite{Fernandez-Pola2019}.

On the other hand, there are also various ways to incorporate the effect $E_a$ of the antiangiogenic treatments (those that induce tumor cell starvation, by inhibiting the tumor vascular support). They depend on the mechanism used by the different antiangiogenic drugs. For those (such as Angiostatin, Endostatin or TNP-470) seeking to block blood vessel growth, the proposed mathematical model (\cite{Hahnfeldt-etal-1999}) was  

\begin{equation} 
K'(t) = - \lambda_2 K(t) + b V(t) - d K(t)V^{2/3}(t)-E_a(t)K(t).
\label{EK2}
\end{equation}

More recently (\cite{Poleszczuk-etal-2015}), a modification was introduced to improve the modeling of the effect of drugs that inhibit angiogenic stimulation (such as Bevacizumab). It takes the following form:

\begin{equation} 
K'(t) =  - \lambda_2 K(t) + \frac{b V(t)}{1+E_a(t)} - d K(t)V^{2/3}(t).
\label{EK3}
\end{equation}

For a specific expression of the cytotoxic effect of the treatment, we refer to the Emax model  
\begin{equation} E_c(t)=\dfrac{k_1c_c(t)}{k_2+c_c(t)}
\label{EEC}
\end{equation} 
with $c_c(t)$ denoting the concentration of the cytotoxic drug. Positive parameters $k_1$ and $k_2$ are experimentally estimated and represent, respectively, the maximum effect of the drug on the body and the half maximal concentration (concentration producing 50\% of the maximum effect, also known as $EC50$).  Anyway, here we will manage with a generic function $E_c$ because it simplifies the presentation and makes the conclusions applicable to more general frameworks.

Similarly, concerning the antiangiogenic treatment effect, we have in mind the expression \begin{equation} E_a(t) = e_ac_a(t),
\label{EECA}
\end{equation} 
where $c_a(t)$ denotes the concentration of the antiangiogenic drug and $e_a > 0$ is another parameter representing the drug impact. In the literature it has been estimated for models (\ref{EK2}) and (\ref{EK3})  in \cite{Hahnfeldt-etal-1999} and \cite{Poleszczuk-etal-2015}, respectively.

In summary, when no therapy is being applied, combining equations (\ref{EV1}) and (\ref{EK1}), we arrive to the following mathematical model for the system formed by the tumor and the surrounding vasculature:

\begin{equation} 
\left\{ \begin{aligned}
& V'(t) = - \lambda_1 V(t) \log \left( \frac{V(t)}{K(t)} \right), & V(0) = V_0,\\
& K'(t) = - \lambda_2 K(t) + b V(t) - d K(t)V^{2/3}(t), & K(0) = K_0.
\end{aligned}
\right.
\label{EVK}
\end{equation}

On the other hand, when one or both therapies are being applied, we will focus on the following ODE system:

\begin{equation} 
    \left\{ \begin{aligned}
V'(t) &= - \lambda_1 V(t) \log \left( \frac{V(t)}{K(t)} \right)\left(1-E_c(t)\right), \ V(0) = V_0, \\
K'(t) &= - \lambda_2 K(t) +\frac{b V(t)}{1+E_a(t)} - d K(t)V^{2/3}(t), \ K(0) = K_0,
\end{aligned}
\right.
\label{EVKwT}
\end{equation} 
although the ODE for $K$ could be changed by (\ref{EK2}) and our main conclusions would remain the same. 

Let us remark that for a cytotoxic non-antiangiogenic drug, $E_c(t) \in \mathbb{R}^+$ and $E_a(t) \equiv 0$ and when the drug only has effect on the vasculature but not on the tumor, then $E_c(t) \equiv 0$ and $ E_a(t) \in \mathbb{R}^+$. Of course, there  exist also   cytotoxic antiangiogenic treatments and then $E_c(t), E_a(t) \in \mathbb{R}^+$. As it is usual in the literature we are denoting $\mathbb{R}^+=(0,+\infty),$ 
$\mathbb{R}^-=(-\infty,0)$  and $\mathbb{R}^+_0=[0,+\infty)$.

When $b > \lambda_2$, it is well-known that the system (\ref{EVK}) has a unique critical point and in the sequel, it will be denoted by  $(V_c,V_c),$ where $V_c = \left(\frac{b-\lambda_2}{d}\right)^\frac{3}{2}.$ 

Let us recall some classical results about the behavior of the solutions of the system (\ref{EVK}), depending on the values of the parameters appearing there. 

\begin{thrm}
Let us assume that $\lambda_1, b, d \in \mathbb{R}^+$ and $\lambda_2 \in \mathbb{R}^+_0$. Given any initial condition $(V_0,K_0) \in (\mathbb{R}^+)^2,$ the Cauchy problem (\ref{EVK}) has a unique solution 
$(V,K) \in (C^\infty(\mathbb{R}^+_0))^2$.
Moreover, $(V(t),K(t)) \in (\mathbb{R}^+)^2$ for each $t \in \mathbb{R}^+_0$ and it is satisfied 
\begin{itemize}
    \item[$a)$] If $b > \lambda_2$, then $\displaystyle \lim_{t \to +\infty} \left(V(t),K(t)\right) = ( V_c,V_c)$.
    \item[$b)$] If $b \leq \lambda_2$, then $\displaystyle \lim_{t \to +\infty} \left(V(t),K(t)\right) = (0,0)$.
\end{itemize}
\label{T1}
\end{thrm}

\begin{proof}
These results are proven (for instance) in \cite{Donofrio-Gandolfi2004} (see Propositions 1 and 3).
Let us emphasize that item $a)$ is equivalent to say  that $(V_c,V_c)$ is a globally asymptotically stable critical point in $(\mathbb{R}^+)^2$ for system (\ref{EVK}). Let us point out that, when $b \leq \lambda_2$, technically $(0,0)$ is not a critical point of the system, because $\lim_{(V,K)\rightarrow (0,0)} V\log \left( \frac{V}{K} \right)$ does not exist; but if it were, in this case we would say that  $(0,0)$ is also globally asymptotically stable.
\end{proof}

Figure~\ref{fig:Grafica_lambda1pos}  shows the two different behaviors of the trajectories.
In biological terms, case $a)$ represents the situation when the tumor is growing towards its critical lethal size $V_c$ (here being, $V_c = 17269 \ mm^3$) and case $b)$ when the tumor is shrinking towards its elimination.

\begin{figure}
    \centering
    \includegraphics[width=0.9\textwidth]{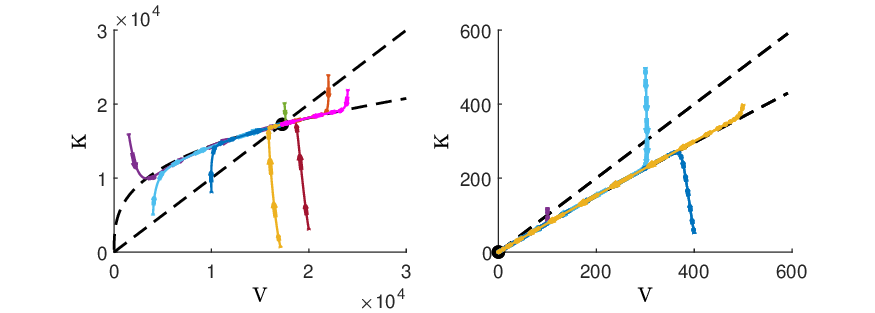}
    \caption{Trajectories for $\lambda_1>0$, $b > \lambda_2$ and  $b \leq \lambda_2$ in the left and right graphs, respectively. The  parameter values are those of  Table~\ref{tablaparametros1} except that for the right one,  $\lambda_2=1.7$. }
    \label{fig:Grafica_lambda1pos}
\end{figure}

Now, let us study the less common case when $\lambda_1$ is negative which will be of interest for dealing with cytotoxic treatments:

\begin{thrm}
Let us assume that $\lambda_1 \in \mathbb{R}^-,$ $b,d \in \mathbb{R}^+$ and $\lambda_2 \in \mathbb{R}^+_0$. Given any initial condition $(V_0,K_0) \in (\mathbb{R}^+)^2$ (different from the critical point $(V_c,V_c)$ when $b > \lambda_2$), the Cauchy problem (\ref{EVK}) has a unique solution 
$(V,K) \in (C^\infty(\mathbb{R}^+_0))^2$.
Moreover, $(V(t),K(t)) \in (\mathbb{R}^+)^2$ for each $t \in \mathbb{R}^+_0$ and it is satisfied one of the following possibilities:
    \begin{itemize}
    \item[$a)$] $\displaystyle \lim_{t \to +\infty} V(t) =0$, or
    \item[$b)$] $\displaystyle \lim_{t \to +\infty} (V(t), K(t)) = (+\infty,+\infty)$.
    \end{itemize}
\label{T2}
\end{thrm}

\begin{proof}
First, the existence and uniqueness of a local solution is a straightforward consequence of any classical theorem about the topic, taking into account that the functions defining the ODEs of system (\ref{EVK}) are $C^1$ in the domain $(\mathbb{R}^+)^2$. After extending (on the right) this solution as far as possible, let us denote by $[0,T)$ its maximal interval of existence. Since we are dealing with a dynamical system, we will look at the phase plane $(V,K)$ and the nullclines $K=V$ and $K = \gamma(V)$, where $\gamma(V) = bV/(\lambda_2+dV^{\frac{2}{3}})$.

We will distinguish two cases: 
\begin{itemize}
\item [i)] When $b > \lambda_2$, the nullclines intersect at the critical point $(V_c,V_c)$, dividing the first quadrant of the phase plane into four regions: 
\[ Q_1 =\{(V,K) \in (\mathbb{R}^+)^2 : \max\{V,\gamma(V)\} < K\}, \]
\[ Q_2 =\{(V,K) \in (\mathbb{R}^+)^2 : V < K < \gamma(V)\}, \]
\[ Q_3 =\{(V,K) \in (\mathbb{R}^+)^2 :  \gamma(V) < K < V\}, \]
\[ Q_4 =\{(V,K) \in (\mathbb{R}^+)^2 : K < \min\{V,\gamma(V)\}\}, \]
together with their boundaries.

By observing the signs of $V'(t)$ and $K'(t)$ in each region, we can easily deduce how the trajectories of the solutions evolve over time. For an initial point in $Q_1$ the corresponding solution will move to the left and down and there are two possibilities: $V(t)$ will tend to $0$ as $t \rightarrow T^-$ (and simultaneously, $K(t)$ will tend to some unknown $K_f > 0$) or the trajectory will enter from $Q_1$ to $Q_3$ and $V(t)$ and $K(t)$ will tend to $+\infty$ as $t \rightarrow T^-$, because  trajectories in $Q_3$ always move to the right and must pass to  $Q_4$. For initial points in $Q_2$, the corresponding solution will move to the left and up and once the solution touches region $Q_1,$ $V(t)$ will tend to $0$ as $t \rightarrow T^-$, unable to return to $Q_2$. Inside $Q_4$ any trajectory moves to the right and up. So, there are two options: it reaches $Q_2$ and later on $Q_1$ and then it behaves as we have explained before or $V$ and $K$ tend to $+\infty$  without being able to leave $Q_4$. The graph  on the left of Figure~\ref{fig:Grafica_lambda1neg} displays some trajectories in different cases.

\begin{figure}
    \centering
    \includegraphics[width=0.9\textwidth]{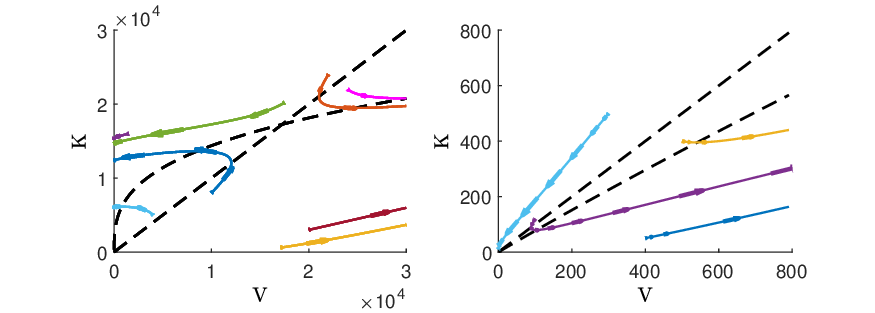}
    \caption{Trajectories for $\lambda_1=-2<0$, $b > \lambda_2$ and  $b \leq \lambda_2$ in the left and right graphs, respectively. The  initial points and the parameter values are those of  Figure~\ref{fig:Grafica_lambda1pos} except for $\lambda_1$.}
    \label{fig:Grafica_lambda1neg}
\end{figure}

In this case we can show that $(V_c,V_c)$ is an unstable critical point for  system (\ref{EVK}).

\item [ii)] When $b \leq \lambda_2$, the first quadrant of the phase plane is divided just in three regions and their boundaries, because now the nullclines do not intersect. They are:  
\[ \tilde{Q}_1 =\{(V,K) \in (\mathbb{R}^+)^2 : V < K \}, \] 
\[\tilde{Q}_2 =\{(V,K) \in (\mathbb{R}^+)^2 :  \gamma(V) < K < V\} \ \ \mbox{and} \]
\[ \tilde{Q}_3 =\{(V,K) \in (\mathbb{R}^+)^2 : K < \gamma(V)\}. \]

Again, we can easily deduce how the trajectories of the solutions evolve over time. For an initial point in $\tilde{Q}_1$ the corresponding solution will move to the left and down and again there are two possibilities: $V(t)$ will either tend to $0$ as $t \rightarrow T^-$ or the trajectory will change to $\tilde{Q}_2$ and $V(t)$ will tend to $+\infty$ as $t \rightarrow T^-$, because the trajectories in $\tilde{Q}_2$ always move to the right and down and, after passing into $\tilde{Q}_3$, they will continue to move to the right and up without leaving this region. Therefore $V$ must tend to $+\infty$. The graph on the right of Figure~\ref{fig:Grafica_lambda1neg} shows some trajectories in different cases.\end{itemize}

It remains to prove that $T = +\infty$. To that end, let us observe that when $0 < K_{min} \leq K(t) \leq K_{max} $ and $V(t) > 0$, 
since $\lambda_1 < 0,$ we derive
\[  - \lambda_1 V(t) \log \left( \frac{V(t)}{K_{max}} \right) \leq  - \lambda_1 V(t) \log \left( \frac{V(t)}{K(t)} \right) \leq  - \lambda_1 V(t) \log \left( \frac{V(t)}{K_{min}} \right).  \]
Using a classical comparison theorem for the solution of the Cauchy problems corresponding to previous right hand sides, we arrive to 
\[ 0 < K_{max} \exp{\left(\log\left(\frac{V_0}{K_{max}}\right)\exp{(-\lambda_1 t)}\right)} 
\leq V(t) \leq \]
\[ \leq K_{min} \exp{\left(\log\left(\frac{V_0}{K_{min}}\right)\exp{(-\lambda_1 t)}\right)}.\]
Above inequalities imply that $V$ can neither be zero nor have a vertical asymptote for a finite time $T$, when $K$ is bounded at $T$. Furthermore, the possibility that both $V$ and $K$ tend to $+\infty$ at the same finite $T$, can be excluded taking into account that it should be satisfied $V(t) > K(t) > 1$ when $t \in [\tilde{T},T),$ for a certain time $\tilde{T} < T,$ then 
\[ V'(t) \leq  - \lambda_1 V(t) \log \left(V(t)\right), \ \ t \in [\tilde{T},T),  \]
and hence, 
\[ \log\left(\log{(V(t))}\right) \leq -\lambda_1(t-\tilde{T})+\log\left(\log{(V(\tilde{T}))}\right), \ \ t \in [\tilde{T},T). \]
Consequently, $T = +\infty$. 
All above possibilities can be summarized as in the statement of Theorem \ref{T2}.
\end{proof}

\begin{rmrk}
\begin{itemize} 
\item [i)] When $b > \lambda_2$ and $V_0 = K_0 = V_c$, it is clear that $(V(t),K(t)) \equiv (V_c,V_c)$.
\item [ii)] According to our model, it follows that the tumor will never disappear completely, from a mathematical point of view.
\end{itemize} 
\end{rmrk}

\section{Constant drug effects}

For simplicity, we will first consider the  system (\ref{EVKwT}) with constant values for the drug effects: $E_c (t) \equiv E_c \in \mathbb{R}^+_0$ and $E_a(t)  \equiv E_a \in \mathbb{R}^+_0$. Hence, we can apply the conclusions of Theorems \ref{T1} and \ref{T2} to this situation directly, with 
$\lambda_1 (1-E_c)$ taking the place of $\lambda_1$ and $\frac{b}{1+E_a}$ instead of $b$. Consequently, the new value for the critical point becomes $(\hat{V_c},\hat{V_c})$ with $\hat{V_c}=\left( \dfrac{b-\lambda_2(1+E_a)}{d (1+E_a)} \right)^\frac{3}{2}$. 

Therefore, when $E_c < 1$, Theorem \ref{T1} with $\lambda_2 \neq 0$ gives the following results: 
\begin{itemize}
    \item[a)] If $0 \leq E_a < (b-\lambda_2)/\lambda_2$, then $\displaystyle \lim_{t \to +\infty} \left(V(t),K(t)\right) = ( \hat{V}_c,\hat{V}_c)$.
    
    \item [b)] If $E_a \geq \max\{(b-\lambda_2)/\lambda_2,0\},$ 
then $\displaystyle \lim_{t \to +\infty} \left(V(t),K(t)\right) = \left(0,0\right).$
\end{itemize}

Again, previous case $a)$ can be interpret as the tumor volume stabilizes around $\hat{V}_c$.  Note that $\hat{V}_c \leq V_c$ and it is smaller the higher $E_a$. This option can be of interest for palliative treatments. Concerning case $b)$ (where the tumor is theoretically being eliminated), let us point out that the required antiangiogenic drug concentration $c_a$ could be very high and not achievable in the practice. 

Alternatively, when $E_c > 1$, applying Theorem \ref{T2} we arrive to one of the following opposite possibilities: 
\[ \displaystyle \lim_{t \to +\infty} V(t) = 0 \ \ \mbox{or} \displaystyle \lim_{t \to +\infty} V(t) = + \infty, \]
assuming $(V_0,K_0) \neq (\hat{V}_c,\hat{V}_c)$. This information does not clarify the situation for each particular case. Hence, it is natural to classify the initial conditions according to the behavior of the corresponding solution of system (\ref{EVKwT}). In what follows, we will focus on the zone where $V \in (0,V_c),$ because this is the area of medical interest, assuming that when $V \geq V_c$ the patient will have a fatal outcome.
In the case where $V$ tends to $0$ as $t \rightarrow + \infty$ (without touching $V_c$) we will say that $(V_0,K_0)$ is a normal point for system (\ref{EVKwT}) with the given cytotoxic  therapy  because it is working satisfactorily (as desired), while in the case where $V$ tends to $+\infty$ as $t \rightarrow + \infty$ or it reaches $V_c$ at some finite time, we will say that $(V_0,K_0)$ is an abnormal point for the system (\ref{EVKwT}) with that therapy, because it is working the opposite of expected. Let us emphasize that the above definitions apply in the case of constant effects. Generally speaking, we will say that an initial point is {\bf abnormal} for a given treatment (or combination of treatments) when the tumor grows more after its application over a long period of time than without any treatment.  Otherwise, we will say that it is  {\bf normal} for that therapy.

In Figure~\ref{Grafica_REGIONES1} we show the normal region (in cyan) and the abnormal region (in red) inside the rectangle $0 < V, K \leq 16850$ (recall that $V_c = 17269$), obtained for therapies with constant effects $E_c = 1.5$ and $E_c = 1.77$, in both cases together with $E_a = 0$. 
First, we check the crucial fact that there exist abnormal initial points with $V_0 \in (0,V_c)$.
Second, that these regions depend on the parameter values and change when considering different cytotoxic drug effect. As we can see, the abnormal zone seems to grow with $E_c$.  
Let us emphasize that the boundary between the normal and abnormal regions is quite blurry and its precise contour will depend on the mesh used to discretize the rectangle. 

\begin{figure}
    \centering    \includegraphics[width=0.9\textwidth]{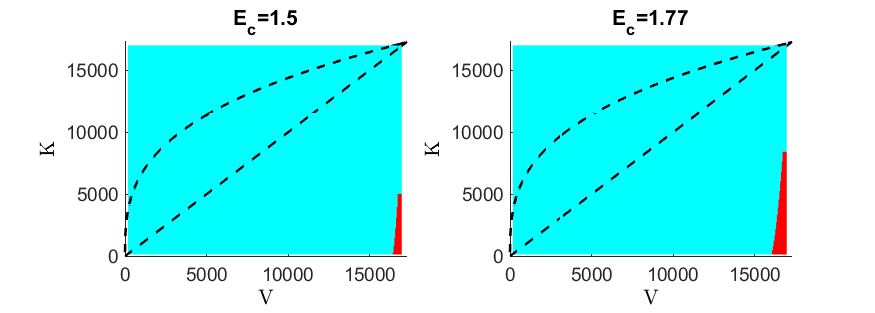}
    \caption{Initial points leading to $V(t) \longrightarrow 0$ in cyan and those for which $V(t) \longrightarrow +\infty$ in red. The  parameter values are those of Table~\ref{tablaparametros1} with $E_a=0$.}
    \label{Grafica_REGIONES1}
\end{figure}

Considering the biology of the tumor-vasculature system, it seems difficult to imagine any reason why this system could ever be in the abnormal zone, even simply below the diagonal $K=V$. It is true that under normal conditions, the tumor will have a smaller volume than the surrounding vasculature in order to be supported by it, but also that in certain special cases (e.g. after some previous treatment) this may be not the case (see for example \cite[Figure 10]{Donofrio-Gandolfi2010}).

\section{Piecewise constant drug effects}

Let us show that very different situations can occur when treating the same tumor with different piecewise constant therapies, by means of some simple examples with the following  parameters taken from \cite{Poleszczuk-etal-2015}: 

\begin{table}[ht]
    \centering
    \begin{tabular}{|c|c|c|c|}
    \hline
    $\lambda_1$ &$\lambda_2$ & $b$ &$d$  \\
    $day^{-1}$$ $& $day^{-1}$$ $ & $day^{-1}$ & $day^{-1}mm^{-2}$  \\
    \hline $0.0741$&$0.0021$&$1.3383$&$0.002$   \\
         \hline
    \end{tabular}
    \caption{Parameter values for the ODE system (\ref{EVKwT}).}
    \label{tablaparametros1}
\end{table}

Two initial conditions have been chosen for our numerical experiments:

\begin{itemize}
    \item [$a)$] With the initial condition $(V_0,K_0) = (16300,1000)$ (both in  $mm^3$, denoted by $(P_1)$ in the sequel) we have considered the following cases:
    \begin{itemize}
    \item [$a_1)$] If no therapy is applied ($(E_c,E_a) \equiv (0,0)$), the  tumor volume at $t = 30$ days will be $V(30) \approx 16701 \ mm^3$.  Moreover, $V(t) \rightarrow V_c = 17269 \ mm^3$ as $t \rightarrow + \infty$, as expected a priori, according to Theorem \ref{T1}$-a)$. On the left side of Figure~\ref{fig:EXP10Ec0Ea}) we see the evolution of the untreated tumor from $(P_1)$. It may be surprising that in the first few days the tumor shrinks. This is because the volume of the surrounding vasculature is too small to support the current tumor volume. In the right part of the figure, starting from a vasculature slightly larger than the tumor volume, this  initial shrinkage no longer appears.

\begin{figure}[htp]
    \centering
    \includegraphics[width=4.3in]{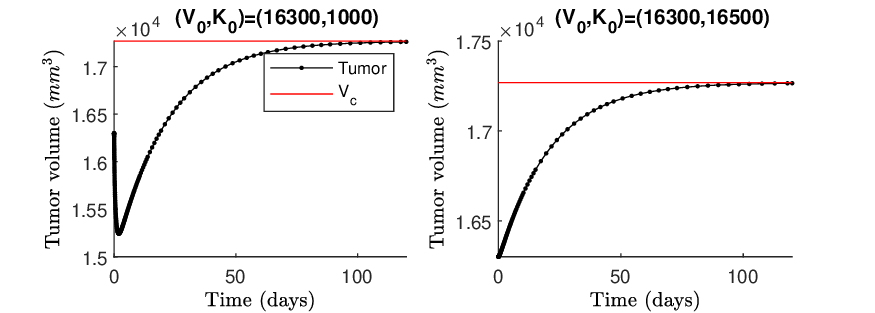}
    \caption{Evolution of the untreated tumor size, using the parameter values of  Table~\ref{tablaparametros1}, for the initial points $(P_1)$ and $(P_2)$ in the left and right graphs, respectively.}
    \label{fig:EXP10Ec0Ea}
\end{figure}

    \item [$a_2)$] If the particular cytotoxic treatment  $(E_c,E_a) \equiv (1.5,0)$ is applied, it can be seen in the plot of the top left  of Figure~\ref{fig:EXP1cteEc0Ea} that "the tumor volume may even increase at the initial stages of treatment", as noted in \cite{Castorinaetal2022}. Again, this appears to be due to the presence of dead material within the tumor. After one month, $V(30) \approx \ 16626 \ mm^3$, slightly lower than in the case of the untreated tumor. Keeping the same therapy, the tumor tends to  disappear: for instance, $V(250) \approx  0.31 \ mm^3$, see Theorem \ref{T2}$-a)$.  

\begin{figure}[htp]
    \centering
    \includegraphics[width=4.3in]{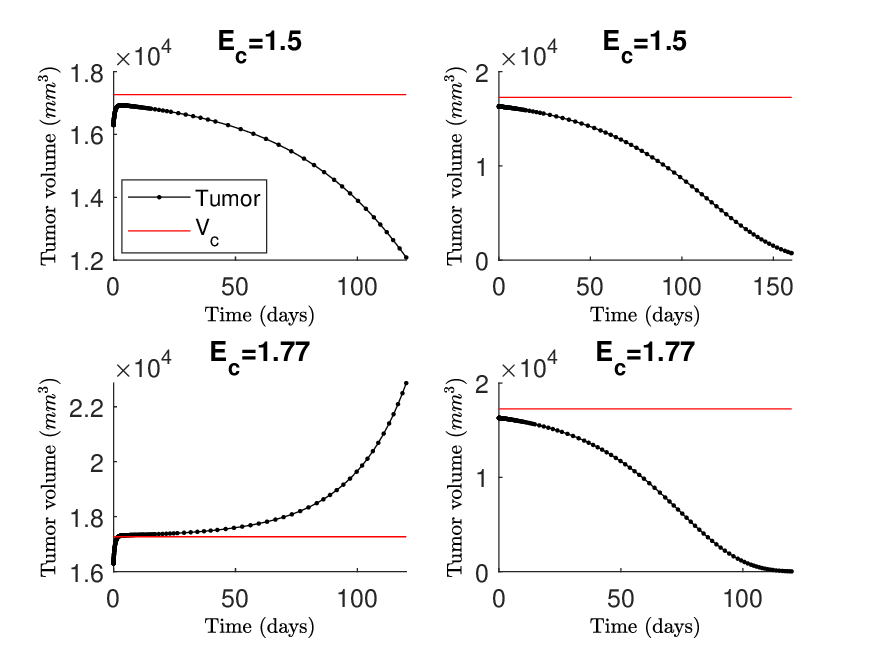}
    \caption{Some results for constant drug effects with $E_a = 0$, parameter values of Table~\ref{tablaparametros1}, with the initial condition $(P_1)$ on the left and  $(P_2)$  on the right.}
    \label{fig:EXP1cteEc0Ea}
\end{figure}

    \item [$a_3)$] Naively arguing, we could try to speed up tumor shrinkage by applying a stronger cytotoxic treatment, such as $(E_c,E_a) \equiv (1.77,0)$ at the bottom   of Figure~\ref{fig:EXP1cteEc0Ea},  but paradoxically the plot on the left  shows that $V(t) \rightarrow +\infty$ as $t \rightarrow + \infty$, displaying that $(P_1)$ is an abnormal point for this level of therapy (see Theorem \ref{T2}$-b)$). In fact, the tumor volume will reach the fatal level $V_c$ on the second day of treatment and therapy is supposed to end at that time.

   \item [$a_4)$] One way to prevent the above situation is to delay the start of cytotoxic therapy, when this is possible. For example, in the plot on the left of Figure \ref{fig:EXP12cetapasOSOLOa} we can see the evolution of tumor volume with the strategy
\[ (E_c(t),E_a(t)) \equiv \left\{\begin{array}{rl} (0,0), & \mbox{ when } t \in [0,2), \\ (1.77,0), & \mbox{ when } t \in [2,30],\end{array}\right. \]
that causes, after $30$ days, $V(30) \approx 12087 \ mm^3$.

 \begin{figure}[htp]
    \includegraphics[width=4.3in]{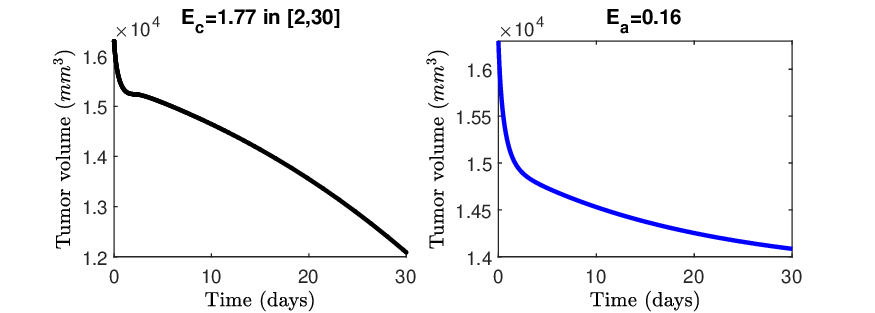}
    \caption{Evolution of tumor size  using $(P_1)$ and   values of  Table~\ref{tablaparametros1} for cytotoxic treatment with initial delay on the left and only antiangiogenic treatment on the right.}
    \label{fig:EXP12cetapasOSOLOa}
\end{figure}
    
    \item [$a_5)$] Alternatively, an antiangiogenic treatment can be applied for some days to normalize the tumor: for instance, with $(E_c,E_a) \equiv (0,0.16)$, the  tumor volume at $t = 30$ days will be $V(30) \approx 14084 \ mm^3$ (see on the right plot  of 
    Figure~\ref{fig:EXP12cetapasOSOLOa}).   It is important to note that during this process  abnormal behavior will never appear since $E_c \equiv 0$.
\end{itemize}

 \item [$b)$] The graphics on the right column  of 
 Figure~\ref{fig:EXP1cteEc0Ea} show the results for the initial condition  $(V_0,K_0) = (16300,16500)$ (both in  $mm^3$, denoted by $(P_2)$ in the sequel).
 For both treatments ($(E_c,E_a) \equiv (1.5,0)$ on the top right and $(E_c,E_a) \equiv (1.77,0)$ on the bottom right) we have that $V(t) \rightarrow 0$ when $t \rightarrow + \infty,$ but the healing process is faster in the second case 
($V(120) \approx 26 \ mm^3$) than in the first ($V(120)
 \approx 5538 \ mm^3$), as one could expect. This means that $(P_2)$ is a normal point for these  therapies.
\end{itemize}

The basic idea to avoid paradoxical results consists of using some procedure to change the initial condition from abnormal to normal before applying the cytotoxic treatment. This can be achieved due to the unstable nature of the critical point $(V_c,V_c)$. As we have just seen in $a_4)$, a useful procedure may  simply be to allow  the tumor-vasculature system evolve without treatment from an abnormal initial point for a short period of time and then apply the cytotoxic treatment. 
There is also a known technique to improve the outcome of cytotoxic therapies that can be applied here.  It consists of first applying an antiangiogenic treatment to achieve a  transient normalization of the vasculature (as in $a_5)$), followed by cytotoxic treatment (see  \cite{Jain2001, Jain2005, Jain-etal-2011}). Here, it is crucial to determine the appropriate dosage of the antiangiogenic treatment. 

As indicated in the Introduction, there are several biological theories trying to explain the abnormal behavior of the tumor after the cytotoxic treatment. The main one is related with the debris generated by the therapies. Our macroscopic approach is compatible with this theory taking into account that during the tumor regression ``the volume dynamics may be influenced by the removal of the abundant necrotic material produced by cell death", see \cite{Donofrio-Gandolfi2004}.

\section{Time-dependent drug effects}

In practice, drug concentrations usually vary over time within the body, so it seems impossible to achieve constant (or piecewise constant) drug effects. Therefore, we have to deal with the case of time-dependent drug effects, that is, the system (\ref{EVKwT}) with piecewise continuous functions $E_c$ and $E_a$ depending on the  concentrations of the corresponding cytotoxic and antiangiogenic treatments. It is clear that jump discontinuities will occur at the times at which the drugs are administered. Let us start by establishing that the problem (\ref{EVKwT}) is also well-posed in this case.
As usual, we will denote by $W^{1,\infty}[0,T]$ the Sobolev space of all functions in $L^\infty(0,T)$ having first order weak derivative (in the distributional sense) also belonging to
$L^\infty(0,T)$. It is well known that $W^{1,\infty}[0,T]$ can be identified with $C^{0,1}[0,T]$, the space of Lipschitz continuous functions in $[0,T]$, after a possible redefinition on a set of zero measure, see \cite[Theorem 5, pg. 269]{Evans-1998}.

\begin{thrm}
Let us assume that $\lambda_1, b, d, T \in \mathbb{R}^+$ and $\lambda_2 \in \mathbb{R}^+_0$. Let $E_c, E_a$ be nonnegative piecewise continuous functions in $[0,T]$ with a finite number of jump discontinuities in that interval. Given any initial condition $(V_0,K_0) \in (\mathbb{R}^+)^2,$ then the Cauchy problem (\ref{EVKwT}) has a unique solution 
$(V,K) \in (W^{1,\infty}[0,T])^2$ and moreover $(V(t),K(t)) \in (\mathbb{R}^+)^2$ when $t \in [0,T].$
\label{T3}
\end{thrm}

\begin{proof}
Let us denote by $\{t_1,\ldots,t_N\}$ the jump discontinuities points for $E_c$ or $E_a$ in $[0,T]$.
First, the local existence and uniqueness of a solution pair of (\ref{EVKwT}) defined in some interval $[0,h]$, with $h \in (0,t_1),$ is a consequence of any classical theorem for this topic, because clearly the functions defining the system (\ref{EVKwT}) are continuous in $[0,t_1)\times (\mathbb{R}^+)^2$ and $C^1$ with respect to $V$ and $K$ in the same domain. We can show that $h = t_1$ noticing that if $0 \leq V(t) \leq \tilde{C},$ $t \in [0,h],$ we have
\[ -(\lambda_2+d\tilde{C}^{2/3})K(t) \leq K'(t) \leq b \tilde{C}, \] and therefore
\[ 0 < K_0\exp{\left(-t(\lambda_2+d\tilde{C}^{2/3})\right)} \leq K(t) \leq b \tilde{C} t + K_0. \]
Hence, $V$ can not be zero at $h$ jointly with $K$. On the other hand, the possibility that both $V$ and $K$ tend to $+\infty$ at $h$ simultaneously can also be  excluded arguing as at the end of Theorem \ref{T2}, when $K < V$ and using that $E_c$ is bounded from above; in the region $V < K,$ it is a consequence of the fact that 
\[ K'(t) \leq b V(t) \leq b K(t). \]
Then, we can estimate the solution from above and from below, concluding that $V$ and $K$ exist in $[0,t_1]$. This reasoning can be followed for the other intervals $[t_i,t_{i+1}],$ $i = 1, \ldots,N-1$ and $[t_N,T]$, pasting the solution at the points $t_i$ and taking into account once more that they cannot be zero or tend to infinity at any finite time.\end{proof}

The behavior of the non-autonomous system trajectories (\ref{EVKwT}) will largely depend on the expression of the effect functions $E_c$ and $E_a$, but (of course) it could definitely be much more complex than in the case of constant effects, see 
 Figure~\ref{fig:Grafica_NAtrajectories}.

\begin{figure}
    \centering
    \includegraphics[width=0.9\textwidth]{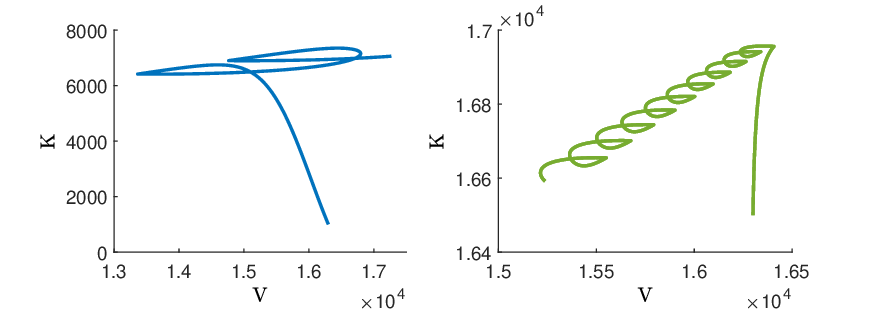}
    \caption{Trajectories with data from the second case  of Table~\ref{tablaresultados2} except that for the one on the left,  $\lambda_2=1.7$. The initial points are $(P_1)$ and $(P_2)$ on the left and right graphs, respectively.}
    \label{fig:Grafica_NAtrajectories}
\end{figure}

Now, let us present some numerical experiments related to the framework introduced in \cite{Fujita-etal-2007}. In that work, the antiangiogenic drug Bevacizumab was used in combination with the cytotoxic drug Paclitaxel to treat a head and neck squamous cell carcinoma on mice. Since it is known that Bevacizumab follows linear pharmacokinetics (PK), see \cite{Poleszczuk-etal-2015}, the simplest Cauchy problem can be used to describe the evolution of its concentration:
\begin{equation}
\left\{ \begin{array}{l}
c'(t) = -\lambda c(t) + \displaystyle\sum^{N}_{i=1}\sigma d_{i}\delta (t-t_{i}), \ \ \ t \in \ [0,T],\\
             c(0) = 0 .\\
             \end{array}
\right.
\label{pCauchy}
\end{equation}
The coefficient $\lambda$ is the clearance rate and is related to the half-life of the drug ($t_{1/2}$) by the expression $\lambda= log(2)/t_{1/2}$. The second term on the right hand of the equation (\ref{pCauchy})  depends on both the specific drug and the way it is administered. We assume that $N$ doses, $\{d_i\}_{i=1}^N$, will be given at $N$ dosage times, $\{t_i\}_{i=1}^N $, such that $ 0 \leq t_1 < \ldots < t_{N} < T,$
being $T$ the final observation time. The coefficient $\sigma$ is determined on the basis of drug-, patient- and tumor-specific parameters, such as the bioavailability of the drug (fraction of the administered drug that reaches the systemic circulation) and the patient's weight. Finally, $\delta(t-t_{i})$ denotes the Dirac delta distribution concentrated at $t_i$.

Concerning the cytotoxic drug, although Paclitaxel is known to  follow a nonlinear PK (\cite{Sparreboom-etal-1996}), some works have continued to use a linear model (see \cite{Benzekry-etal-2017}, \cite{Simeoni-etal-2004}). Keeping in mind that our purpose here is purely illustrative, for the sake of simplicity, we will also use the linear PK (\ref{pCauchy}). Parameters associated with the antiangiogenic drug will be indicated with the subscript ``$a$”, while those associated with the cytotoxic one will be indicated with the subscript ``$c$”.

To compare different treatments, the criterion ``Area Under the Curve" (AUC), which measures the degree of exposure to a drug, is often used. Mathematically, it can be calculated as follows:
\begin{equation}\label{AUC}
    AUC = \displaystyle \int_0^{T} c(t) \ dt.
\end{equation} Let us recall the explicit formulation for the concentration of each drug and AUC in the discrete case:
\begin{prpstn}
Under previous notations, the solution of the Cauchy problem \eqref{pCauchy} is given by
\begin{equation}
        c(t)=\left\{ \begin{array}{ll}
            0 ,& \hspace{0.3cm} \mbox{if}\ \  t \in [0,t_{1}),\\
            \sigma e^{-\lambda t}e^{\lambda t_{1}}d_{1} , & \hspace{0.3cm} \mbox{if}\ \ t \in [t_{1},t_{2}),\\
            \sigma e^{-\lambda t}(e^{\lambda t_{1}}d_{1} + e^{\lambda t_{2}}d_{2}) , & \hspace{0.3cm} \mbox{if}\ \ t \in [t_{2},t_{3}),\\
            \vdots & \hspace{0.3cm} \vdots\\
            \sigma e^{-\lambda t}(e^{\lambda t_{1}}d_{1} + \ldots +e^{\lambda t_{N-1}}d_{N-1}) , & \hspace{0.3cm} \mbox{if}\ \ t \in [t_{N-1},t_{N}),\\
            \sigma e^{-\lambda t}(e^{\lambda t_{1}}d_{1} + \ldots +e^{\lambda t_{N}}d_{N}) , & \hspace{0.3cm} \mbox{if}\ \ t \in [t_{N},T).\\
        \end{array}
    \right.
    \label{c(t)}
    \end{equation}
    Furthermore, \begin{equation} \label{AUCnocd}
     AUC = \dfrac{1}{\lambda}\sum_{i=1}^{N} \sigma d_i (1-e^{\lambda (t_{i}-T)}) .
\end{equation}
\end{prpstn}

\begin{proof}
Expression (\ref{c(t)}) is well known in the literature (see for instance \cite{Fernandez-Pola-SainzPardo2022}). Formula (\ref{AUCnocd})  for AUC is just a straightforward consequence of it. \end{proof}

For simplicity and convenience, the administration times will be taken equispaced ($ \Delta = t_{i}-t_{i-1}$ for $i= 2,\ldots, N$, starting at $t_1 = 0$ or $t_1 > 0$, depending on the case), and all the doses equal ($d_i=\overline{d},$ for $i=1,\ldots, N$). 

In what follows, we will use the expressions (\ref{EEC}) and (\ref{EECA}) for the drug effects $E_c$ and $E_a$, with $c_c(t)$ and $c_a(t)$ denoting the concentrations of the cytotoxic and antiangiogenic drug, respectively.
Moreover, $\Delta_c, \overline{d}_c$ and $N_c$ will denote the interval  between dosage times, the individual dose and the number of cytotoxic doses (resp. $\Delta_a, \overline{d}_a$ and $N_a$ for the antiangiogenic drug).

To compare the results of the model with constant drug effect and the discrete treatment (see Table \ref{tablaresultados1}  and Figure~\ref{fig:EXP4cteDISCRETeffect} ), using the Mean Value Theorem for integrals, we will identify  
\begin{equation}
E_a = \frac{e_a}{T}\int_0^T c_a(t)  dt = \frac{e_a}{T}AUC_a = \dfrac{\sigma_a e_a \overline{d}_a}{T\lambda_a}\left(N_a-\sum_{i=1}^{N_a} e^{\lambda_a (t_{i}-T)}\right),
\label{Ea}
\end{equation}
see (\ref{EECA}) and (\ref{AUCnocd}); also, 
\begin{equation}
E_c = \frac{k_1}{T}\int_0^T \frac{c_c(t)}{k_2+c_c(t)}  dt.
\label{Ec}
\end{equation}
More precisely, we will manage the approximated  discrete version of this identity, 
\begin{equation}
E_c =  \frac{N_c k_1}{\lambda_c T} \log\left(1 + \frac{\bar{d}_c \sigma_c}{k_2}\right),
\label{EcAp}
\end{equation}
that holds under the hypothesis (see \cite{Fernandez-Pola-SainzPardo2022})  
\begin{equation} 
\sigma_c \bar{d}_c e^{-\lambda_c \Delta_c} \ll k_2. 
\label{Hyp}
\end{equation} 

For the cytotoxic discrete treatment we have utilized the parameters described in Table \ref{tablaparametros1}, together with the ones included in Table \ref{tablaparametros2}, just for illustrative purposes, following this procedure:
\begin{itemize}
\item [i)] The clearance rate parameter $\lambda_c$ is taken from \cite{Benzekry-etal-2017}. 
\item [ii)] $\sigma_c = F \cdot w/V_d,$  where $F=0.31$ is the bioavailability of Paclitaxel (\cite{Gelderblom-etal-2002}), $w=1/50 \ kg$ is the average weight of a mouse and $V_d = 18.5 \ ml$ is the Paclitaxel volume of distribution in this framework (\cite{Benzekry-etal-2017}). Hence, $\sigma_c = 0.000335$ $kg/ml$.
\item [iii)] For $k_2$ we have taken into account that ``the estimated EC50 value is very low compared to the peak plasma concentration of paclitaxel,...according to the PK model", see \cite{Mo-etal-2014}, and we have used the same percentage cited there for an average dose of $10 \ mg/kg$: so, $k_2 = 0.022$ $mg/ml$.
\item [iv)] Looking at (\ref{EcAp}) and noticing that we can write  
\begin{equation}
\frac{c_c(t)}{k_2+c_c(t)}  =  \frac{\tilde{c}_c(t)}{\frac{k_2}{\sigma_c}+\tilde{c}_c(t)},
\label{E100}
\end{equation}
where $\tilde{c}_c$ is the concentration corresponding to the cytotoxic drug with $\sigma_c=1$, numerically it is more convenient to handle the ratio $k_2/\sigma_c$ as a single parameter rather than managing $k_2$ and $\sigma_c$ separately. 

\item [v)] Coefficient $k_1$ can be derived from  
(\ref{EcAp}), once the constant $E_c$ value is obtained by qualitatively fitting our model (\ref{EVKwT}) (with constant effect) to the experimental data from \cite{Fujita-etal-2007}.  
\end{itemize}

\begin{table}[ht]
    \centering
    \begin{tabular}{|c|c|c|}
    \hline
     $\lambda_c$  & $k_1$ & $k_2/\sigma_c$ \\
   $day^{-1}$ & & $mg/kg$ \\
   \hline $5.55$& $115.12$ & $65.67$\\
   \hline
    \end{tabular}
    \caption{Parameters for the cytotoxic drug.}
    \label{tablaparametros2}
    \end{table}

Let us stress the different results for the therapies, depending on the initial point. Note that the hypothesis (\ref{Hyp}) is satisfied in all cases and the total dose administered during a treatment is equal to $160$ $mg/kg$. Final results are reported at $t = 31$ days, although the treatment is administered over $[0,30]$, to include the effect until more than $99.6\%$ of the cytotoxic drug has been eliminated from the body. Recall that $(P_1)$ refers to $(V_0,K_0) = (16300,1000)$ and $(P_2)$ to 
$(V_0,K_0) = (16300,16500)$.

\begin{table}[ht]
    \centering
    \begin{tabular}{|c|c|c|c|c|c||c|c|c|}
    \hline
     \multicolumn{6}{|c||}{Discrete treatment}  & \multicolumn{3}{c|}{Constant effect treatment}  \\
     \hline
    \multicolumn{3}{|c|}{}  & \multicolumn{2}{c|}{$(P_1)$} & \multicolumn{1}{c||}{$(P_2)$}  & \multicolumn{1}{|c|}{}  & \multicolumn{1}{c|}{$(P_1)$} & \multicolumn{1}{c|}{$(P_2)$}  \\
    \hline
    \hline
         $\Delta_c$  & $N_c$   & $\overline{d}_c$  & $T$   & $V(T)$  & $V(31)$  &  $E_c$ & $V(31)$   &   $V(31)$   \\
          $day$    &      & $mg/kg$ & $day$ & $mm^3$  & $mm^3$  &        & $mm^3$    &   $mm^3$   \\
    \hline $5$     & $6$  & $26.67$ & $1$   & $V_c$   & $15707$ & $1.37$ &   $16429$ & $15613$  \\
    \hline $3$     & $10$ & $16$    & $1$   & $V_c$   & $15455$ & $1.46$ &   $16543$ & $15381$ \\
    \hline $2$     & $15$ & $10.67$ & $1$   & $V_c$   & $15294$ & $1.51$ &   $16630$ & $15235$ \\
    \hline $1$     & $30$ & $5.33$  & $1$   & $V_c$   & $15094$ & $1.57$ &   $16748$ & $15062$ \\
    \hline $0.5$   & $60$ & $2.67$  & $1$   & $V_c$   & $14985$ & $1.6$ &   $16822$ & $14963$ \\
    \hline $0.33$  & $90$ & $1.78$  & $31$ & $17154$  & $14954$ & $1.61$ &   $16850$ & $14928$ \\
         \hline
    \end{tabular}
        \caption{Results for the cytotoxic discrete treatment and constant drug effect starting at $t_1 = 0$.}
        \label{tablaresultados1}
\end{table}

\begin{figure}[htp]
    \centering
    \includegraphics[width=4.3in]{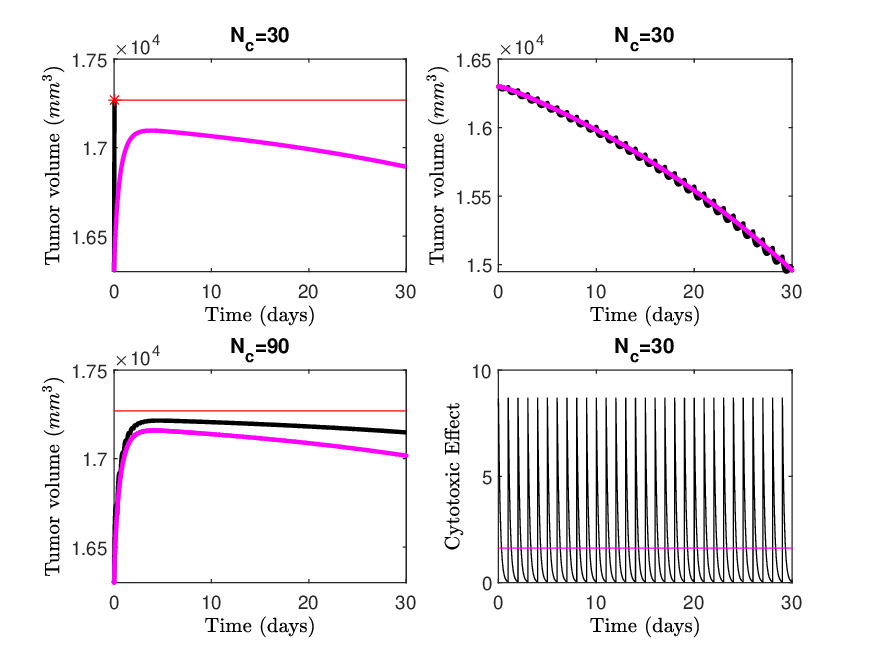}
    \caption{Results with cytotoxic discrete treatments (in black) and constant effect treatment (in magenta) of  
 Table~\ref{tablaresultados1} for the initial point $(P_1)$ on the left and $(P_2)$ on the right.}
    \label{fig:EXP4cteDISCRETeffect}
\end{figure}

\begin{table}[ht]   
    \centering
    \begin{tabular}{|c|c|c|c|c|c|c|}
    \hline 
     \multicolumn{4}{|c|}{} & \multicolumn{1}{c|}{$(P_1)$} & \multicolumn{1}{c|}{$(P_2)$}  \\
    \hline
    \hline
        $\Delta_c$  & $N_c$   & $\overline{d}_c$  & $Total \ dose$ & $V(31)$ & $V(31)$ \\
          $day$   &       & $mg/kg$ & $mg/kg$     & $mm^3$   & $mm^3$  \\
    \hline $5$    & $6$   & $26.67$ & $160$       & $13378$  & $15413$\\
    \hline $3$    & $10$  & $16$    & $160$       & $12956$  & $15238$\\
    \hline $2$    & $14$  & $10.67$ & $149.33$    & $13258$  & $15414$\\
    \hline $1$    & $28$  & $5.33$  & $149.33$    & $12968$  & $15290$\\
    \hline $0.5$  & $55$  & $2.67$  & $146.67$    & $12962$  & $15297$\\
    \hline $0.33$ & $82$  & $1.78$  & $145.78$    & $12971$  & $15304$\\
         \hline
    \end{tabular}
    \caption{Results for the cytotoxic discrete treatment starting at $t_1 = 3$.}
    \label{tablaresultados2}
\end{table}

 \begin{figure}[htp]
    \centering
    \includegraphics[width=4.3in]{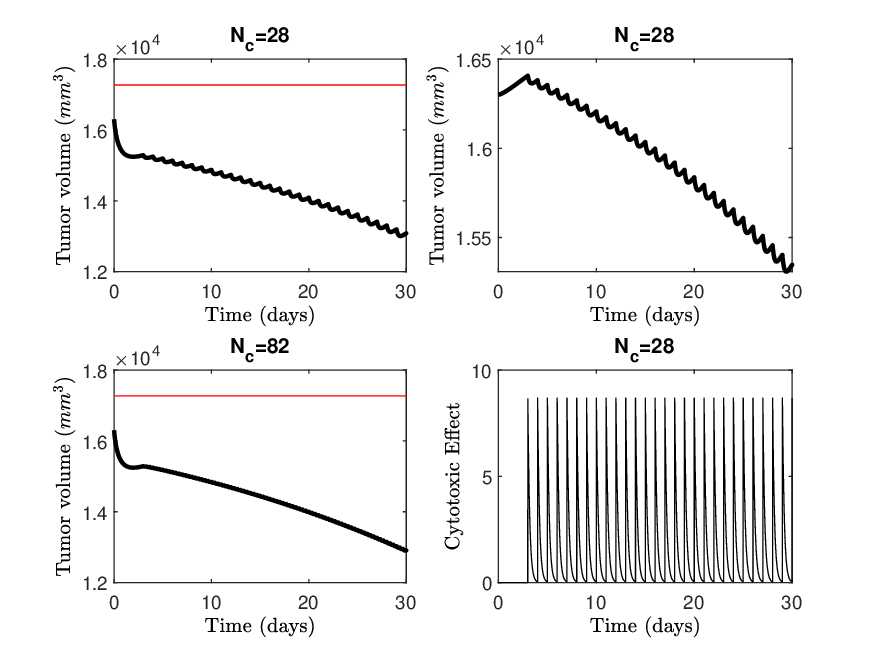}
    \caption{Results with cytotoxic discrete treatments  with three-day delay (in black)  of  
 Table~\ref{tablaresultados2} for the initial point $(P_1)$ on the left and $(P_2)$ on the right.}
    \label{fig:EXP4cteDISCRETeffectretraso}
\end{figure}

\begin{itemize}
    \item [i)] For the initial point $(P_1)$,  when the cytotoxic discrete treatments start at $t_1 = 0$,  the fourth and fifth columns of Table \ref{tablaresultados1} show that most of them can not be completed, because almost immediately (the first day) the tumor volume reaches the fatal level $V_c$ (in red in 
    Figure~\ref{fig:EXP4cteDISCRETeffect}). The exception is the case with $90$ doses, see  bottom left of 
    Figure~\ref{fig:EXP4cteDISCRETeffect}. Note that again we can see the paradoxical behavior that the tumor volume during cytotoxic treatment is larger than without any treatment (see item $a_1)$ of the previous section and the portion corresponding to the first $30$ days of the plot on the left of Figure~\ref{fig:EXP10Ec0Ea}). On the other hand, when the cytotoxic discrete treatments start at $t_1 = 3$ (see fifth column of Table \ref{tablaresultados2}),  surprisingly all the treatments are completed and the tumor volumes at the final time $T=31$ account for a greater reduction than $24\%$ of the volume of the $90$ dose case, even though the total dose is lower than $160$ $mg/kg$  (see fourth column of Table \ref{tablaresultados2}).  It seems that the waiting days before applying the first cytotoxic dose allows the vasculature of the system to normalize, at least partially. Table \ref{tablaresultados2} and  
    Figure~\ref{fig:EXP4cteDISCRETeffectretraso} consider  cytotoxic treatments corresponding to the same inter-dose spacing and same individual dose as in  
    Table~\ref{tablaresultados1} and  
    Figure~\ref{fig:EXP4cteDISCRETeffect}, now adjusting the number of doses to the interval $[3,30]$ to take into account the delay in the onset of treatments (delay that is evident in the fourth graph of 
    Figure~\ref{fig:EXP4cteDISCRETeffectretraso}). 
    
    \item [ii)] For the initial point $(P_2)$, the situation is clearer and it does not depend on the starting time for the treatment. The volume at the final time  decreases with the number of doses $N_c$ in Table \ref{tablaresultados1}. In  Table \ref{tablaresultados2} the reduction is smaller, but this is due to the fact that the total dose also decreases almost $9\%$. In both tables the  metronomic chemotherapy appears to be the preferred option. Moreover, comparing the results in Table~\ref{tablaresultados1} between the discrete treatment and the constant effect treatment, we see in the sixth and ninth columns that the final tumor volumes corresponding to the constant effect are very close, but slightly lower than those for the discrete treatment. Figure~\ref{fig:EXP4cteDISCRETeffect}  shows on the right    a similar evolution of the tumor volumes obtained with both treatments for one of the cases in the Table~\ref{tablaresultados1}.
\end{itemize}

Let us recall that, according to the National Cancer Institute (NCI) of USA (\cite{NCI}), metronomic chemotherapy is ``the treatment in which low doses of anticancer drugs are given on a continuous or frequent, regular schedule (such as daily or weekly), usually over a long time." There are several medical reasons for this choice such as minimal adverse effects and no need of prolonged drug-free breaks. Among all possible options for discrete therapies appearing in the literature, many authors prefer  metronomic chemotherapy. We have seen that there are also mathematical reasons to recommend it: in a previous study, we pointed out the superiority of metronomic chemotherapy in many situations (\cite{Fernandez-Pola-SainzPardo2022}), using a different approach.

Concerning the antiangiogenic drug effect $E_a$, we will use the expression (\ref{EECA}) with $c_a$ denoting the concentration of the drug (in our case,  Bevacizumab). Parameters associated with this drug are taken from \cite{Poleszczuk-etal-2015}; let us mention that, as it is commonly done in the literature (see for instance \cite{{Hahnfeldt-etal-1999}} and \cite{Poleszczuk-etal-2015}), although it is not usually made explicit, we have once again preferred to handle $\sigma_a$ jointly with the drug coefficient $e_a$. 

\begin{table}[ht]  
    \centering
    \begin{tabular}{|c|c|}
    \hline
     $\lambda_a$  & $\sigma_a e_a$\\
   $day^{-1}$ &$kg/mg$ \\
    \hline $0.0799$&$0.4755$ \\
         \hline
    \end{tabular}
    \caption{Parameters for the antiangiogenic drug Bevacizumab (\cite{Poleszczuk-etal-2015}).}
    \label{tablaparametros3}
\end{table}

Firstly, we have applied three antiangiogenic discrete treatments with the same total dose, starting at $t_1 = 0$ and finishing at $t_N = 28$ days, obtaining the results shown in Table \ref{tablaresultados2c}. For both initial points $(P_1)$ and $(P_2)$ these are quite similar and stable. Moreover, we have observed that there is no benefit in postponing the onset to some $t_1 > 0$. Recall that here the abnormal behavior does not appear because there is no cytotoxic treatment. We have included the final volumes at $31$ and also at $72$ days, because the antiangiogenic effect will continue to act for a long period of time after the end of treatment: in the case of Bevacizumab, since its half-life is about $9$ days, $97\%$ of the drug will not disappear until $44$ days later. Here, there is a huge difference with the cytotoxic drug, because the half-life of Paclitaxel is about $3$ hours.  It should be noted that during the antiangiogenic treatments the volume of the vasculature will remain mainly smaller than the tumor volume in all cases, but will be larger from a few days after the end of the treatments. Therefore, this seems an appropriate strategy to normalize the tumor for subsequent cytotoxic treatment; depending on whether this should be applied (closer to $t=31$ or $t=72$ days), a more or less metronomic strategy can be chosen. 

\begin{table}[ht]
    \centering
    \begin{tabular}{|c|c|c||c|c||c|c|}
    \hline 
     \multicolumn{3}{|c|}{} & \multicolumn{2}{c|}{$(P_1)$} & \multicolumn{2}{c|}{$(P_2)$}  \\
    \hline
    \hline
    $\Delta_a$ &$N_a$ & $\overline{d}_a$ & $V(31)$ & $V(72)$  & $V(31)$ & $V(72)$ \\
    $day$ &  & $mg/kg$  & $mm^3$& $mm^3$  & $mm^3$& $mm^3$ \\
        \hline $7$ & $5$ & $6.4$ & $2600$ & $5726$ & $2698$ & $5764$ \\
        \hline $4$ & $8$ & $4$ & $2492 $ & $5802$ & $2570$ & $ 5835$ \\
         \hline $2$ & $15$ & $2.13$ & $2442$ & $5865$ & $2505$ & $5892$ \\
        \hline
    \end{tabular}
    \caption{Results for antiangiogenic discrete treatments starting at $t_1 = 0$ with the same total dose.}
    \label{tablaresultados2c}
\end{table}

In this work, we have used cytotoxic or antiangiogenic drugs separately, but it is clear that their combination could be applied in some way to optimize the outcome. There are some papers that address this issue, such as \cite{Fujita-etal-2007} and \cite{Benzekry-etal-2017}; in our opinion, this question deserves a specific mathematical study, taking into account other important external factors, such as side effects. To get an idea of the sensitivity of the problem with respect to the initial condition, we present here the following example: by administering two antiangiogenic doses of $\overline{d}_a = 4$ $mg/kg$ at days $t_1 = 0$ and $t_2 =4$ to $(P_2)$
followed by a cytotoxic dose of $\overline{d}_c = 26. 67$ $mg/kg$ on day $7$, the tumor reaches the critical level very quickly ($V(8) \approx V_c$). This means that $(P_2)$ becomes an abnormal point for this combination. Let us recall that $(P_2)$ behaved as a normal point for all treatments presented earlier in this paper.

We have also performed some numerical experiments to measure the effect of noise on the administration times, since in practice it is very difficult to be completely accurate to this aspect. For this purpose, we have added a truncated Gaussian noise with zero mean and variance $\Delta_c/30$ lying at $[-\Delta_c/10,\Delta_c/10]$ to the dosing times of the discrete cytotoxic treatment for the initial point $(P_1)$ and repeated $1000$ runs, calculating the mean of the final volumes $V(31)$ as well as its range. Comparing the results of Tables \ref{tablaresultados2} and \ref{tablaresultados5b}, we observe that the numerical results are very stable with respect to these perturbations and that again the metronomic treatments perform slightly better, achieving almost the same final volumes than in Table \ref{tablaresultados2}.

\begin{table}[ht]   
    \centering
    \begin{tabular}{|c|c|c|c|c|c|}
    \hline
    $\Delta_c$ &$N_c$ &  $\overline{d}_c$ &  $Mean$ $V(31)$  & $Interval \ for$ $V(31)$ \\
    $day$ &  & $mg/kg$  & $mm^3$ & \\
     \hline $5$ & $6$ & $26.67$ & $13384$ & $[13264,13591]$ \\
    \hline $3$ & $10$ & $16$ & $12957$ & $[12905,13025]$ \\
       \hline $2$ & $14$ & $10.67$ & $13258$ & $[13232,13287]$ \\
     \hline $1$ & $28$ & $5.33$ & $12968$ &  $[12959,12977]$\\
    \hline $0.5$ &$55$ & $2.67$ & $12962$ &  $[12960,12965]$\\
    \hline $0.33$ &$82$& $1.78$& $12971$ & $[12970,12972]$\\
         \hline
    \end{tabular}
    \caption{Results of $1000$ experiments adding a truncated Gaussian noise to the  dosage times.}
    \label{tablaresultados5b}
\end{table}

Finally, it is also interesting to measure the sensitivity of each parameter with respect to noise. To this end, for each parameter we have performed $1000$ experiments by adding a uniformly distributed random variable with values in the range of $\pm 10\%$ of its initial value, keeping the other parameters unchanged. The results are collected in the Table \ref{tablaresultados6}, with  the mean and range of the $1000$ values of $V(31)$ using the cytotoxic discrete treatment for the initial point $(P_1)$ with $N_c=82$, $\Delta_c = 0.33$ days and dose $\overline{d}_c = 1.78$ $mg/kg$, starting at $t_1 = 3$ days. It is worth mentioning that we have preferred to handle the ratio $k_2/\sigma_c$ as a single parameter instead of handling $k_2$ and $\sigma_c$ separately, see (\ref{E100}). Clearly, the first consequence of this analysis is that the spontaneous loss of
vasculature $\lambda_2$ has no significant influence on the tumor size (this was predicted by \cite{Hahnfeldt-etal-1999} and confirmed by \cite{Poleszczuk-etal-2015}). 
Furthermore, from our analysis we conclude that the most influential parameters are the stimulation due to the tumor $b$ and the endogenous inhibition $d$. Compared with \cite{Poleszczuk-etal-2015}, here they seem to have a somewhat more relevant role. The other four parameters (the tumor growth parameter $\lambda_1$ and those related with the cytotoxic drug: $\lambda_c$, $k_1$ and $k_2/\sigma_c$) seem to have a similar weight on the final tumor volume.  

\begin{table}[ht]
    \centering
    \begin{tabular}{|c|c|c|c|}
    \hline
    $Perturbed \ \ parameter$ &  $Mean$ $V(31)$  & $Interval \ for$ $V(31)$ \\
     & $mm^3$ & \\
     \hline $\lambda_1$ & $12976$ & $[12457,13449]$ \\
     \hline $\lambda_2$ & $12971$ & $[12967,12975]$ \\
     \hline $b$ &  $13103$ & $[11155,15323]$ \\
     \hline $d$ & $12943$ & $[10787,15293]$\\
     \hline $\lambda_c$ & $12902$ &  $[11958,13674]$\\
     \hline $k_1$ & $12967$ & $[12056,13744]$\\
     \hline $k_2/\sigma_c$ & $12903$ & $[11967,13668]$\\
    
         \hline
    \end{tabular}
     \caption{Results of $1000$ experiments adding $\pm 10\%$ uniform noise to the indicated parameter.}
    \label{tablaresultados6}
\end{table}

All the previous computations were performed on a $3.9$ GHz Core $i5-8265U$ machine, with $8$ GB RAM running under the $64$-bit version of Windows $11$ and MATLAB $R2024b$.

\section{Conclusions}

From the above theorems and numerical experiments, we can draw the following conclusions:
\begin{itemize}

\item[1)] Cytotoxic chemotherapy transforms the tumor-vasculature system into an unstable one when it is applied at a sufficiently high concentration. Then, in normal cases (the majority), tumors shrink in size or slow down their growth, as is well known, but there are also abnormal situations whose volume (paradoxically) increases faster with cytotoxic chemotherapy treatment than without any treatment. There is extensive medical literature mentioning this paradoxical phenomenon (see \cite{Chang-etal-2019}, \cite{Dalterioetal2020}, \cite{Haak-etal-2021} and \cite{Sulciner-etal-2018} among others). In this work, we have presented a mathematical explanation for this paradox, using a classical minimally parameterized model together with the Norton-Simon hypothesis. As far as we know, models studied by other authors do not include these paradoxical cases, because they assume instead the log-kill hypothesis (see  for instance \cite{Urszula2015}). Hence, the Norton-Simon hypothesis is a key point in our argumentation. On the other hand, our main results remain unaffected if we change the Gompertz ODE to the logistic one, for instance.

\item[2)]  The set of abnormal cases varies with the treatments used and seems to be small, but expands with the effect of the cytotoxic chemotherapy treatment. It is mainly located around the area where the tumor volume is very large and the vasculature is small. The boundary between this and the set of normal cases is quite blurred and it can be approximately determined by numerical techniques in the model.  It seems crucial to distinguish whether a situation is abnormal before applying a given cytotoxic chemotherapy treatment. Two initial points having the same tumor volume with different (but close) initial vasculature can lead to very different results when treated with the same therapy. As we have seen, a tumor can move from being abnormal to be normal, changing the employed  treatments (see $(P_1)$ in Tables \ref{tablaresultados1} and \ref{tablaresultados2c}) ; but we should be very careful because just the opposite situation can also happen: a tumor moving from being normal to be abnormal, see for instance $(P_2)$ in Table \ref{tablaresultados1} and $(P_2)$ with the  combined treatment mentioned in the above section.

\item[3)] Fortunately, there are different ways to avoid the abnormal situations. One is simply to delay the onset of treatment, if this is possible. In this case, the size of the tumor vasculature should be monitored until it reaches an adequate size. To our knowledge this strategy has not been reported in the literature, but some references suggest that monitoring the size of the vasculature is an important tool in oncology, see \cite{Poleszczuk-etal-2015}.
A second one is to apply a previous antiangiogenic treatment to achieve a transient normalization of the vasculature and later a cytotoxic treatment. This technique has been reported in various medical sources (see \cite{Jain2001, Jain2005, Jain-etal-2011}) and our system also reflects this fact.

\end{itemize}

{\bf Acknowledgement.}

The first author was supported by MCIU/AEI/10.13039/501100011033/ under research project  PID2023-147610NB-I00.

\end{document}